\documentclass[11pt]{article}

\usepackage[margin=1in]{geometry}
\usepackage{amsmath,amssymb,amsthm,mathtools,bm}
\usepackage{booktabs}
\usepackage{graphicx}
\usepackage{subcaption}
\usepackage{tikz}
\usetikzlibrary{arrows.meta,positioning}
\usepackage[numbers,sort&compress]{natbib}
\usepackage[hidelinks]{hyperref}
\usepackage[nameinlink,capitalise]{cleveref}
\usepackage{xcolor}
\usepackage{microtype}
\usepackage{placeins}
\usepackage{float}

\newtheorem{theorem}{Theorem}
\newtheorem{lemma}{Lemma}
\newtheorem{proposition}{Proposition}
\newtheorem{corollary}{Corollary}
\newtheorem{remark}{Remark}
\newtheorem{definition}{Definition}

\newcommand{\R}{\mathbb{R}}
\newcommand{\Z}{\mathbb{Z}}
\newcommand{\T}{\mathbb{T}}
\newcommand{\Sph}{\mathbb{S}}
\newcommand{\E}{\mathbb{E}}
\newcommand{\Prob}{\mathbb{P}}
\newcommand{\Tr}{\operatorname{Tr}}
\newcommand{\dd}{\mathrm{d}}
\newcommand{\Id}{\mathrm{I}}
\newcommand{\SO}{\mathrm{SO}}
\newcommand{\SU}{\mathrm{SU}}
\newcommand{\diag}{\operatorname{diag}}
\newcommand{\norm}[1]{\left\lVert #1\right\rVert}

\title{Non-Abelian Spin Counting of Ordered Stochastic Trajectories: Reentrant Finite-Time Chern Numbers}
\author{Yangyang Du\thanks{Department of Mathematics, University of Michigan, Ann Arbor, Michigan, USA}}
\date{\today}

\begin{document}
\maketitle

\begin{abstract}
Conventional full counting statistics assigns commuting phases to integrated stochastic currents and therefore resolves net transport but not, in general, the temporal ordering of events associated with different cycles. We introduce a non-Abelian counting construction in which crossings of two fundamental cycles rotate an auxiliary spin about different axes. The resulting ordered trajectory statistic has an exact finite-dimensional evolution equation for its first moment. The two rotation angles form a counting torus, and whenever the mean spin is nonzero its normalized direction defines a map $\mathbb T^2\to\mathbb S^2$, equivalently a complex eigenline bundle with a Chern number. Our main analytical result is a Chern--parity correspondence. Reflection symmetry equips this eigenline with a real structure and expresses $C_T\bmod2$ through first Stiefel--Whitney classes on the circles fixed by reflection. When the transverse polarization has no additional zeros along the reflection-fixed circles, these classes reduce to ordinary current-parity statistics at the four high-symmetry counting points. For a five-state nonequilibrium figure-eight network, varying only the observation time produces four polarization-gap closings and the reentrant sequence $C_T=0\to-1\to0\to-1\to0$. Every transition occurs at $(\pi,\pi)$ and coincides with a sign reversal of $\mathbb E[(-1)^{Q_1+Q_2}]$, while the full integer Chern number is obtained independently from the two-dimensional spin texture. Finite observation time can therefore organize a fixed stochastic process into distinct topological sectors of its ordered path ensemble.
\end{abstract}

\section{Introduction}

Full counting statistics (FCS) characterizes fluctuations of time-integrated observables in stochastic and mesoscopic transport. In its standard form, counting fields conjugate to integrated currents are introduced through scalar phase factors or tilted generators \cite{LevitovLeeLesovik1996,BagretsNazarov2003}. For continuous-time Markov processes, this framework underlies large-deviation descriptions of current fluctuations and fluctuation relations for nonequilibrium steady states \cite{LebowitzSpohn1999,AndrieuxGaspard2007}. Finite-time current statistics can contain transient information beyond asymptotic cycle currents \cite{PolettiniEsposito2014}, while stochastic thermodynamics relates trajectory currents to affinities, entropy production, and fluctuation relations \cite{Schnakenberg1976,Seifert2012}.

Suppose a trajectory $\gamma$ carries $d$ integer-valued integrated currents,
\begin{equation}
    \bm Q[\gamma]
    =
    \bigl(Q_1[\gamma],\ldots,Q_d[\gamma]\bigr)
    \in\mathbb Z^d.
    \label{eq::int_val_currents}
\end{equation}
Conventional FCS introduces periodic counting fields $\bm\chi\in\mathbb T^d$ and assigns the scalar weight
\begin{equation}
    \exp\!\left(i\bm\chi\cdot\bm Q[\gamma]\right)
    =
    \exp\!\left(i\sum_{a=1}^d\chi_aQ_a[\gamma]\right).
\end{equation}
This construction is Abelian since the counting factors commute. Consequently, when several cycle currents are counted, the statistic retains the accumulated net currents but generally discards the temporal order in which increments associated with different cycles occurred. Two trajectories can therefore have the same $\bm Q[\gamma]$ and the same Abelian counting weight while representing different ordered sequences of stochastic events.

Topology plays an important role in stochastic counting. Geometric phases of stochastic generating functions generate pump currents \cite{SinitsynNemenman2007,Sinitsyn2009}, and in adiabatically driven networks at low temperature, such currents can exhibit robust integer or fractional quantization \cite{ChernyakKleinSinitsyn2012}. Characteristic-class formulations connect these effects to the homology of stochastic networks and to bundles over spaces of counting and driving parameters \cite{ChernyakKleinSinitsyn2013}. Special degeneracies of counting operators can also produce topological changes in current statistics \cite{ChernyakSinitsyn2010}, and braid-group structures can arise in nonequilibrium counting spectra \cite{RenSinitsyn2013}. In these constructions, however, the counting variables themselves remain Abelian.

Here we replace commuting counting phases by noncommuting rotations of an auxiliary spin. For a two-loop graph, a positive crossing associated with the first fundamental cycle applies $R_x(\alpha)\in\SO(3)$, while a crossing associated with the second applies $R_y(\theta)$; reverse crossings apply the inverse rotations. Since $R_x(\alpha)$ and $R_y(\theta)$ do not commute in general, the final spin depends on the ordered word of counted events. An $A$ event followed by a $B$ event and the reverse ordering can therefore produce different counter states even when both trajectories have identical net cycle currents.

Spin degrees of freedom have long been used as counting detectors. In the Levitov--Lee--Lesovik formulation, an auxiliary spin coupled to electronic current records transferred charge through its precession \cite{LevitovLeeLesovik1996}. Statistics of noncommuting spin observables have also been studied in quantum transport \cite{DiLorenzoCampagnanoNazarov2006}. The construction here is different in that the underlying dynamics are entirely classical and Markovian, while the noncommuting rotations are used specifically to resolve ordering information in stochastic network trajectories. Mathematically, for a two-loop graph with $\pi_1(G)\cong F_2=\langle a,b\rangle$, ordinary cycle-current counting factors through the abelianization $F_2\to H_1(G;\mathbb Z)\cong\mathbb Z^2$, whereas the spin counter uses the family of representations
\begin{equation}
    \rho_{\alpha,\theta}:F_2\longrightarrow\SO(3),
    \qquad
    \rho_{\alpha,\theta}(a)=R_x(\alpha),
    \qquad
    \rho_{\alpha,\theta}(b)=R_y(\theta).
\end{equation}
Related twisted graph constructions with flat unitary line-bundle coefficients provide an Abelian analogue in which edge transports carry $U(1)$ phases \cite{CatanzaroChernyakKlein2013}.

The two periodic rotation angles form a counting torus $\mathbb T^2$. Let $m_T(\alpha,\theta)$ be the finite-time mean spin. Whenever $m_T$ is nonzero on the entire torus, the normalized polarization
\begin{equation}
    n_T(\alpha,\theta)
    =
    \frac{m_T(\alpha,\theta)}{\|m_T(\alpha,\theta)\|}
\end{equation}
defines a map $\mathbb T^2\to\mathbb S^2$. Equivalently, the averaged spin density matrix has a nondegenerate dominant eigenline over the counting torus, and its first Chern number is the degree of $n_T$.  We use the standard language of vector bundles and characteristic classes throughout \cite{Husemoller1994,MilnorStasheff1974}. This topology belongs to the finite-time counter state rather than to the physical Markov generator itself, distinguishing the present construction from topological classifications of Markov generators and configuration-space networks \cite{MuruganVaikuntanathan2017,TangAgudoCanalejoGolestanian2021}.

Our main analytical result connects this non-Abelian Chern number to an ordinary parity statistic without reducing the full counting problem to an Abelian one. Reflection symmetry equips the dominant eigenline with a real structure over the involutive counting torus. Let
\[
C_0=\{\alpha=0\},
\qquad
C_\pi=\{\alpha=\pi\}
\]
denote the two circles fixed by reflection, and let $\ell_0\to C_0$ and $\ell_\pi\to C_\pi$ be the corresponding real eigenline bundles. Their first Stiefel--Whitney classes
\[
w_1(\ell_a)\in H^1(C_a;\mathbb Z_2),
\qquad a\in\{0,\pi\},
\]
measure whether the real eigenline returns with the same or opposite sign after one circuit of $C_a$. Let
\[
[C_a]\in H_1(C_a;\mathbb Z_2)
\]
denote the fundamental homology class of the circle, the pairing
$\langle w_1(\ell_a),[C_a]\rangle\in\mathbb Z_2$ evaluates the cohomology class $w_1(\ell_a)$ on that loop. Such pairing is $0$ for a trivial real line bundle and $1$ for a M\"obius-type line bundle. The Chern number then obeys
\begin{equation}
C_T\equiv
\left\langle w_1(\ell_0),[C_0]\right\rangle
+
\left\langle w_1(\ell_\pi),[C_\pi]\right\rangle
\pmod2.
\end{equation}
If, in addition, the $x$-component of the spin polarization does not vanish along either reflection-fixed circle except at $\theta = 0,\pi$, these $w_1$ invariants are determined by the four high-symmetry polarizations. At $(\alpha,\theta)\in\{0,\pi\}^2$, those polarizations are exactly the conventional parity-generating functions $\mathbb E[(-1)^{\eta_1Q_1+\eta_2Q_2}]$. Thus symmetry relates an Abelian parity observable to the mod-$2$ reduction of the first Chern number of a non-Abelian two-parameter counting family.

For the five-state figure-eight network studied below, changing only the observation time produces four polarization-gap closings and the reentrant sequence
\begin{equation}
    C_T:
    0\longrightarrow-1\longrightarrow0\longrightarrow-1\longrightarrow0.
\end{equation}
Every transition occurs at $(\pi,\pi)$, where
\begin{equation}
Z_T(\pi,\pi)=\mathbb E[(-1)^{Q_1(T)+Q_2(T)}]
\end{equation}
changes sign. The Chern--parity theorem therefore predicts the even/odd alternation of the topological sectors, while an independent two-dimensional degree calculation fixes the integer values and their sign. The finite-time distribution of ordered trajectory words evolves, while the stochastic generator, stationary state, and rotation rule all remain unchanged across the transitions.

The paper is organized around three results. First, we derive an exact linear evolution equation for the mean non-Abelian counter spin. Second, we identify its finite-time Chern number and prove the Chern--parity correspondence described above. Third, we resolve four finite-time topological transitions numerically and compare their critical times with relaxation, accumulated activity, and complete-cycle first-passage scales. Together these results show that observation time can act as a control parameter for topology in the ordered stochastic path ensemble.

\section{Two-loop nonequilibrium Markov network}
\label{sec:model}

Consider the five-state figure-eight graph in \cref{fig:graph}, with oriented cycles
\begin{equation}
    A:0\to1\to2\to0, \qquad B:0\to3\to4\to0.
\end{equation}
The forward and reverse transition rates on cycle $A$ are denoted by $k_{1,+}$ and $k_{1,-}$, while those on cycle $B$ are $k_{2,+}$ and $k_{2,-}$.

\begin{figure}[t]
\centering
\begin{tikzpicture}[
    x=1.45cm,
    y=1.25cm,
    state/.style={circle,draw,minimum size=9mm,inner sep=0pt,font=\normalsize},
    arr/.style={-{Latex[length=2.3mm]},semithick},
    note/.style={font=\small,align=center,fill=white,inner sep=2pt}
]
\node[state] (0) at (0,0) {$0$};
\node[state] (1) at (-3.8,1.55) {$1$};
\node[state] (2) at (-3.8,-1.55) {$2$};
\node[state] (3) at (3.8,1.55) {$3$};
\node[state] (4) at (3.8,-1.55) {$4$};

\draw[arr,bend left=14] (0) to (1);
\draw[arr,bend left=14] (1) to (0);
\draw[arr,bend left=14] (1) to (2);
\draw[arr,bend left=14] (2) to (1);
\draw[arr,bend left=14] (2) to (0);
\draw[arr,bend left=14] (0) to (2);

\draw[arr,bend left=14] (0) to (3);
\draw[arr,bend left=14] (3) to (0);
\draw[arr,bend left=14] (3) to (4);
\draw[arr,bend left=14] (4) to (3);
\draw[arr,bend left=14] (4) to (0);
\draw[arr,bend left=14] (0) to (4);

\node[note] at (-3.8,2.35) {$A$: forward $k_{1,+}$\\reverse $k_{1,-}$};
\node[note] at (3.8,2.35) {$B$: forward $k_{2,+}$\\reverse $k_{2,-}$};

\node[note,anchor=north east] at (-1.25,-0.72)
{$2\to0:\ R_x(\alpha)$\\$0\to2:\ R_x(-\alpha)$};
\node[note,anchor=north west] at (1.25,-0.72)
{$4\to0:\ R_y(\theta)$\\$0\to4:\ R_y(-\theta)$};
\end{tikzpicture}
\caption{Two-loop Markov network and the counting gauge. On cycle $A$ the forward orientation $0\to1\to2\to0$ has rate $k_{1,+}$ and the reverse orientation has rate $k_{1,-}$; cycle $B$ is defined analogously. The auxiliary spin rotates only on the two distinguished closing edges, with inverse rotations on reverse jumps. The counter is passive and does not modify the transition rates.}
\label{fig:graph}
\end{figure}
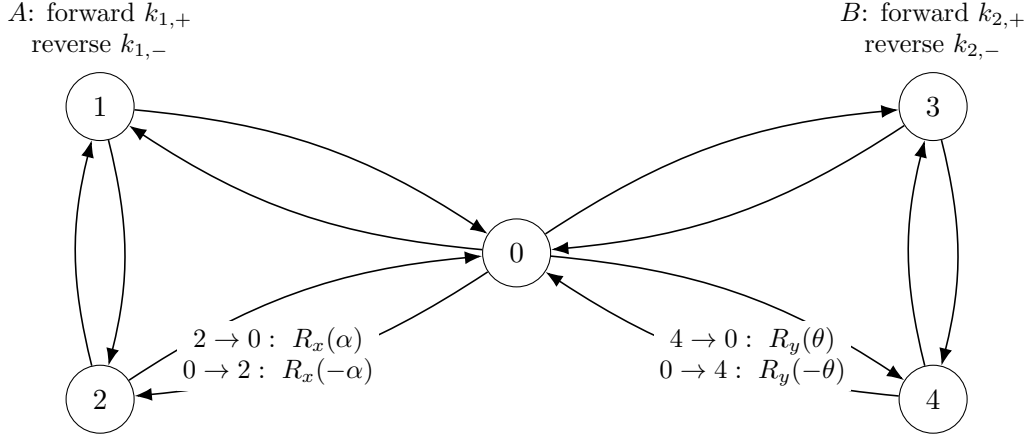

With the column-vector convention, the Markov generator is
\begin{equation}
    L_{ji}=k_{j\leftarrow i}\quad(i\neq j), \qquad L_{ii}=-\sum_{j\neq i}k_{j\leftarrow i}, \qquad \dot p=Lp.
\end{equation}
Although the forward and reverse rates need not be equal, the rate pattern is balanced at every vertex since the total incoming rate equals the total outgoing rate. Consequently, the stationary distribution is uniform,
\begin{equation}
 p_i^{\mathrm{ss}}=\frac15,
 \qquad i=0,\dots,4.
 \label{eq:uniform_stationary}
\end{equation}
When $k_{i,+}\neq k_{i,-}$, this uniform stationary state nevertheless carries a nonzero circulating probability current around cycle $i$.

To represent the two independent cycle currents, we choose the edges $2\leftrightarrow0$ and $4\leftrightarrow0$ as the two chords associated with the fundamental cycles $0\to1\to2\to0$ and $0\to3\to4\to0$, respectively.  We define the corresponding integrated currents as the net numbers of oriented jumps across these edges
\begin{equation}
Q_1(T)=N_{2\to0}(T)-N_{0\to2}(T),
\qquad
Q_2(T)=N_{4\to0}(T)-N_{0\to4}(T).
\label{eq:cycle_counts}
\end{equation}
As in \cref{eq::int_val_currents}, the vector form is
\begin{equation}
    \bm Q(T)=\bigl(Q_1(T),Q_2(T)\bigr)\in\mathbb Z^2.
\end{equation}
Ordinary Abelian full counting statistics would encode these currents through the scalar factor
\begin{equation}
\exp\left[i\bigl(\chi_1Q_1(T)+\chi_2Q_2(T)\bigr)
\right],
\qquad
(\chi_1,\chi_2)\in\mathbb T^2.
\label{eq:abelian_counting_factor}
\end{equation}
The non-Abelian counter introduced below replaces these commuting scalar phases by noncommuting spin rotations.

In the stationary ensemble,
\begin{equation}
J_i=\frac{k_{i,+}-k_{i,-}}5,
\qquad
\E[Q_i(T)]=J_iT,
\qquad
i = 1,2.
\label{eq:cycle_currents}
\end{equation}
The thermodynamic bias around cycle $A$ is measured by its cycle affinity, defined as the logarithm of the product of forward rates divided by the product of reverse rates around that cycle. In this example,
\begin{equation}
\mathcal A_i=3\log\frac{k_{i,+}}{k_{i,-}}, \qquad i = 1,2
\label{eq:affinity}
\end{equation}
are nonzero, confirming that the stationary process violates detailed balance. The auxiliary spin records additional information about the ordered trajectory while leaving these ordinary current and thermodynamic quantities unchanged.

\section{Non-Abelian spin counting}
\label{sec:spin_counting}

Let $s_0=e_z=(0,0,1)^\mathsf T$ and define
\begin{equation}
R_x(\alpha)=
\begin{pmatrix}
1&0&0\\
0&\cos\alpha&-\sin\alpha\\
0&\sin\alpha&\cos\alpha
\end{pmatrix},
\qquad
R_y(\theta)=
\begin{pmatrix}
\cos\theta&0&\sin\theta\\
0&1&0\\
-\sin\theta&0&\cos\theta
\end{pmatrix}.
\label{eq:rotations}
\end{equation}
The jumps $2\to0$ and $4\to0$ carry $R_x(\alpha)$ and $R_y(\theta)$, while the reverse jumps carry their inverses. All other jumps carry the identity.

For a trajectory $\gamma$ with jump sequence $e_1,\ldots,e_N$, let
\begin{equation}
U_\gamma(\alpha,\theta)=R_{e_N}\cdots R_{e_1},
\qquad
s_T(\gamma)=U_\gamma(\alpha,\theta)s_0.
\label{eq:path_ordered_rotation}
\end{equation}
The ordering matters. A single $A$ traversal followed by a $B$ traversal gives
\begin{equation}
s_{AB}=R_y(\theta)R_x(\alpha)e_z
=\begin{pmatrix}
\sin\theta\cos\alpha\\
-\sin\alpha\\
\cos\theta\cos\alpha
\end{pmatrix},
\label{eq:sAB}
\end{equation}
whereas
\begin{equation}
s_{BA}=R_x(\alpha)R_y(\theta)e_z
=\begin{pmatrix}
\sin\theta\\
-\sin\alpha\cos\theta\\
\cos\alpha\cos\theta
\end{pmatrix}.
\label{eq:sBA}
\end{equation}
Both trajectories have $(Q_1,Q_2)=(1,1)$, but $s_{AB}\neq s_{BA}$ for generic angles.

\begin{definition}
For fixed $T$, $\alpha$, and $\theta$, the spin-counting measure is
\begin{equation}
\mu_T^{\alpha,\theta}(B)=\Prob(s_T\in B),
\qquad B\subset\Sph^2.
\label{eq:spin_measure}
\end{equation}
Its first moment is
\begin{equation}
m_T(\alpha,\theta)=\int_{\Sph^2}s\,\mu_T^{\alpha,\theta}(\dd s)=\E[s_T].
\label{eq:mean_spin}
\end{equation}
\end{definition}

The complete spin-counting statistic at fixed $(T,\alpha,\theta)$ is $\mu_T^{\alpha,\theta}$. We focus on its first moment because it is closed under the block dynamics and because it is precisely the Bloch vector of the averaged counter-spin density matrix. Define state-resolved spin moments
\begin{equation}
g_i(t)=\E\!\left[s_t\,\mathbf1_{\{X_t=i\}}\right]\in\R^3,
\qquad
g=(g_0,\ldots,g_4)^\mathsf T\in\R^{15}.
\label{eq:state_resolved_moment}
\end{equation}

\begin{proposition}
The vector $g(t)$ satisfies
\begin{equation}
\dot g=\mathbb L_{\alpha,\theta}g,
\label{eq:block_evolution}
\end{equation}
with $3\times3$ blocks
\begin{equation}
(\mathbb L_{\alpha,\theta})_{ji}
=k_{j\leftarrow i}R_{j\leftarrow i}\quad(j\neq i),
\qquad
(\mathbb L_{\alpha,\theta})_{ii}
=-\left(\sum_{j\neq i}k_{j\leftarrow i}\right)\Id_3.
\label{eq:block_generator}
\end{equation}
For initial distribution $p(0)$ and deterministic spin $s_0$,
\begin{equation}
g_i(0)=p_i(0)s_0,
\qquad
m_T=\sum_{i=0}^{4}g_i(T).
\label{eq:block_initial}
\end{equation}
\end{proposition}

\begin{proof}
Condition on the state at time $t$ and on whether a jump occurs during $[t,t+\dd t]$. A jump $i\to j$ contributes $k_{j\leftarrow i}\dd t\,R_{j\leftarrow i}g_i(t)$ to $g_j(t+\dd t)$, whereas no jump from $i$ contributes $(1-r_i\dd t)g_i(t)$ with $r_i=\sum_{j\neq i}k_{j\leftarrow i}$. Subtracting $g(t)$, dividing by $\dd t$, and taking the limit gives \cref{eq:block_evolution,eq:block_generator}.
\end{proof}

Hence
\begin{equation}
g(T)=e^{T\mathbb L_{\alpha,\theta}}g(0),
\qquad
m_T(\alpha,\theta)=\mathcal E e^{T\mathbb L_{\alpha,\theta}}g(0),
\label{eq:matrix_exponential}
\end{equation}
where $\mathcal E=(\Id_3\ \Id_3\ \Id_3\ \Id_3\ \Id_3)$. Because the entries of $\mathbb L_{\alpha,\theta}$ are trigonometric polynomials, $m_T$ is real analytic and $2\pi$-periodic in both counting angles.

\section{Counter-spin density matrix and Chern number}
\label{sec:chern}

Lift the rotations to $\SU(2)$,
\begin{equation}
U_x(\alpha)=e^{-i\alpha\sigma_x/2},
\qquad
U_y(\theta)=e^{-i\theta\sigma_y/2}.
\label{eq:su2_lifts}
\end{equation}
Starting from $\varrho_0=(\Id_2+e_z\cdot\bm\sigma)/2$, each trajectory yields $\varrho_\gamma=U_\gamma\varrho_0U_\gamma^\dagger$. Averaging gives
\begin{equation}
\varrho_T(\alpha,\theta)
=\E[\varrho_\gamma]
=\frac12\left(\Id_2+m_T(\alpha,\theta)\cdot\bm\sigma\right).
\label{eq:averaged_density_matrix}
\end{equation}
Its eigenvalues are
\begin{equation}
p_\pm=\frac{1\pm\norm{m_T}}2,
\label{eq:density_eigenvalues}
\end{equation}
so the polarization gap is
\begin{equation}
\Delta_{\mathrm{pol}}(T)
=\min_{(\alpha,\theta)\in\T^2}\norm{m_T(\alpha,\theta)}.
\label{eq:polarization_gap}
\end{equation}
Whenever $\Delta_{\mathrm{pol}}(T)>0$, define
\begin{equation}
n_T=\frac{m_T}{\norm{m_T}},
\qquad
P_+=\frac12(\Id_2+n_T\cdot\bm\sigma).
\label{eq:normalized_texture}
\end{equation}
Then $P_+$ defines a complex line bundle over the counting torus \cite{Husemoller1994}.

\begin{proposition}
If $\Delta_{\mathrm{pol}}(T)>0$, then
\begin{align}
C_T
&=\frac{1}{2\pi i}\int_{\T^2}\Tr(P_+\,\dd P_+\wedge\dd P_+)\\
&=\frac1{4\pi}\int_0^{2\pi}\!\int_0^{2\pi}
n_T\cdot(\partial_\alpha n_T\times\partial_\theta n_T)
\,\dd\alpha\,\dd\theta
=\deg(n_T)\in\Z.
\label{eq:degree_formula}
\end{align}
\end{proposition}

\begin{proof}
Insert \cref{eq:normalized_texture} and use $\sigma_a\sigma_b=\delta_{ab}\Id_2+i\epsilon_{abc}\sigma_c$. The resulting curvature is the pullback of the normalized area form on $\Sph^2$, as in the standard two-band Chern construction \cite{SticletEtAl2012}.
\end{proof}
Thus $(\alpha,\theta)$ may be viewed as synthetic Brillouin-zone coordinates and $n_T$ as a polarization texture.

\begin{corollary}
If $C_T$ changes as $T$ varies continuously, then there exists a critical time $T_c$ and counting angles $(\alpha_c,\theta_c)$ such that
\begin{equation}
m_{T_c}(\alpha_c,\theta_c)=0,
\qquad
\Delta_{\mathrm{pol}}(T_c)=0.
\label{eq:gap_closing}
\end{equation}
At that point $p_+=p_-=1/2$, so the dominant eigenline ceases to be defined.
\end{corollary}

\section{Reflection symmetry and the Chern--parity correspondence}
\label{sec:symmetry}

Write
\begin{equation}
m_T(\alpha,\theta)=\bigl(X_T(\alpha,\theta),Y_T(\alpha,\theta),Z_T(\alpha,\theta)\bigr).
\label{eq:XYZ}
\end{equation}
Let
\[
S_\alpha=\diag(1,-1,1),
\qquad
S_\theta=\diag(-1,1,1),
\]
and lift them to $\R^{15}$ by $\widehat S_\alpha=\Id_5\otimes S_\alpha$ and $\widehat S_\theta=\Id_5\otimes S_\theta$.

\begin{lemma}
\label{lem:reflection}
The block generator obeys
\begin{equation}
\mathbb L_{-\alpha,\theta}
=\widehat S_\alpha\mathbb L_{\alpha,\theta}\widehat S_\alpha,
\qquad
\mathbb L_{\alpha,-\theta}
=\widehat S_\theta\mathbb L_{\alpha,\theta}\widehat S_\theta.
\label{eq:generator_reflections}
\end{equation}
For $s_0=e_z$,
\begin{equation}
m_T(-\alpha,\theta)=S_\alpha m_T(\alpha,\theta),
\qquad
m_T(\alpha,-\theta)=S_\theta m_T(\alpha,\theta).
\label{eq:mean_spin_reflections}
\end{equation}
Consequently there are real-analytic functions $F_T,G_T,H_T$ such that
\begin{align}
X_T(\alpha,\theta)&=\sin\theta\,F_T(\cos\alpha,\cos\theta),\label{eq:X_factor}\\
Y_T(\alpha,\theta)&=-\sin\alpha\,G_T(\cos\alpha,\cos\theta),\label{eq:Y_factor}\\
Z_T(\alpha,\theta)&=H_T(\cos\alpha,\cos\theta).\label{eq:Z_factor}
\end{align}
\end{lemma}

\begin{proof}
The rotation matrices satisfy $S_\alpha R_x(\alpha)S_\alpha=R_x(-\alpha)$ and $S_\alpha R_y(\theta)S_\alpha=R_y(\theta)$, with the analogous relations for $S_\theta$. Applying these identities blockwise gives \cref{eq:generator_reflections}. Both reflections fix $e_z$, so \cref{eq:mean_spin_reflections} follows from \cref{eq:matrix_exponential}. Parity and periodicity then imply the factorizations \cref{eq:X_factor,eq:Y_factor,eq:Z_factor}.
\end{proof}

The reflection-fixed circles for $\alpha\mapsto-\alpha$ are $\alpha=0$ and $\alpha=\pi$. Their intersections with the corresponding $\theta$-fixed sets give the four high-symmetry points
\begin{equation}
(0,0),\quad(\pi,0),\quad(0,\pi),\quad(\pi,\pi).
\label{eq:four_corners}
\end{equation}
At these four points both transverse components vanish identically by symmetry, so the mean polarization lies on the $z$ axis. This observation gives the high-symmetry values a direct interpretation in ordinary full counting statistics.

\subsection{Current parity at the high-symmetry points}
\label{sec:parity_high_symmetry}

Define the ordinary two-current characteristic function
\begin{equation}
    \mathcal Z_T(\chi_1,\chi_2)
    =
    \mathbb E\!\left[
        e^{i(\chi_1 Q_1(T)+\chi_2 Q_2(T))}
    \right].
    \label{eq:abelian_fcs_characteristic}
\end{equation}
For an individual trajectory $\gamma$, let
\[
    s_T(\gamma;\alpha,\theta)\in S^2
\]
denote the final Bloch vector obtained by applying, in chronological order, the counter rotations associated with the counted transitions along $\gamma$, starting from $e_z$. The ensemble-averaged polarization introduced above is therefore
\begin{equation}
    m_T(\alpha,\theta)
    =
    \mathbb E_\gamma\!\left[
        s_T(\gamma;\alpha,\theta)
    \right],
\end{equation}
with
\[
    Z_T(\alpha,\theta)
    =
    e_z\cdot m_T(\alpha,\theta).
\]

At $\alpha,\theta\in\{0,\pi\}$, the counter rotations preserve the $z$ axis. In particular,
\begin{equation}
    R_x(\pi)e_z=-e_z,
    \qquad
    R_y(\pi)e_z=-e_z,
\end{equation}
and the inverse $\pi$ rotations act identically on $e_z$.  Hence, for $\eta_1,\eta_2\in\{0,1\}$, the trajectory-level counter state reduces to
\begin{equation}
    s_T\bigl(\gamma;\eta_1\pi,\eta_2\pi\bigr)
    =
    (-1)^{\eta_1 Q_1[\gamma]+\eta_2 Q_2[\gamma]}e_z.
    \label{eq:parity_spin_action}
\end{equation}
The signed currents may be used in the exponent because addition and subtraction agree modulo two.

Taking the ensemble average and using the definitions of $Z_T$ and $\mathcal Z_T$ then gives
\begin{equation}
    Z_T(\eta_1\pi,\eta_2\pi)
    =
    \mathcal Z_T(\eta_1\pi,\eta_2\pi),
    \label{eq:high_symmetry_abelian}
\end{equation}
so that
\begin{align}
    Z_T(0,0)&=1, \\
    Z_T(\pi,0)
    &=P(Q_1\ {\rm even})-P(Q_1\ {\rm odd}), \\
    Z_T(0,\pi)
    &=P(Q_2\ {\rm even})-P(Q_2\ {\rm odd}), \\
    Z_T(\pi,\pi)
    &=P(Q_1+Q_2\ {\rm even})-P(Q_1+Q_2\ {\rm odd}).
    \label{eq:joint_parity}
\end{align}
Thus, at the four high-symmetry points, the trajectory-level non-Abelian counter collapses to a parity sign, and the corresponding ensemble-averaged polarizations coincide exactly with ordinary $\chi=\pi$ parity characteristic functions. Away from these points, the noncommuting counter rotations retain information about the temporal ordering of counted transitions. The special role of $\chi=\pi$ in stochastic counting also appears in \cite{ChernyakSinitsyn2010,RenSinitsyn2013}.

\subsection{Main theorem: reflection reduction of Chern parity}
\label{sec:main_theorem}

In addition to identifying special counting points, the reflection symmetry also relates the parity of the first Chern number to one-dimensional topological data along the reflection-fixed sets.

Assume that $\Delta_{\mathrm{pol}}(T)>0$, so that $m_T(\alpha,\theta)$ is nonzero everywhere and $P_+$ defines a complex eigenline bundle
\[
    L_T=\operatorname{im}P_+
    \longrightarrow \mathbb T^2.
\]
The reflection
\begin{equation}
    \tau(\alpha,\theta)=(-\alpha,\theta)
    \label{eq:real_involution}
\end{equation}
acts on the counting torus. By \cref{lem:reflection}, $P_+(\tau(\alpha,\theta))$ is the complex conjugate of $P_+ (\alpha,\theta)$. Consequently, complex conjugation maps the eigenline over $(\alpha,\theta)$ antilinearly to the eigenline over $\tau(\alpha,\theta)$.

The fixed set of $\tau$ consists of the two circles \begin{equation}
    C_0=\{\alpha=0\},
    \qquad
    C_\pi=\{\alpha=\pi\}.
\end{equation}
Along either fixed circle the projector is real, so its image contains a distinguished real one-dimensional eigenspace. These eigenspaces form real line bundles
\[
    \ell_0\longrightarrow C_0,
    \qquad
    \ell_\pi\longrightarrow C_\pi.
\]
Each such real line bundle is characterized by its first Stiefel--Whitney class
\[
    w_1(\ell_a)\in H^1(C_a;\mathbb Z_2),
    \qquad a\in\{0,\pi\},
\]
which is the obstruction to orientability of the real line bundle \cite{MilnorStasheff1974}.
Let
\[
    [C_a]\in H_1(C_a;\mathbb Z_2)
\]
denote the fundamental homology class of the circle $C_a$. The pairing
\[
    \left\langle w_1(\ell_a),[C_a]\right\rangle
    \in\mathbb Z_2
\]
is the evaluation of the Stiefel--Whitney class on one circuit around the circle. Geometrically, it is $0$ if a continuous real eigenvector can be chosen to return to itself after one loop, and $1$ if it returns with the opposite sign, as for a M\"obius line bundle \cite{MilnorStasheff1974}.

\begin{theorem}[Chern--parity correspondence for the spin counter]
\label{thm:chern_parity}
For any observation time $T$ with $\Delta_{\mathrm{pol}}(T)>0$, the first Chern number satisfies
\begin{equation}
    C_T
    \equiv
    \left\langle w_1(\ell_0),[C_0]\right\rangle
    +
    \left\langle w_1(\ell_\pi),[C_\pi]\right\rangle
    \pmod2.
    \label{eq:chern_sw_relation}
\end{equation}
If, in addition,
\begin{equation}
    F_T(\cos a,\cos\theta)\neq0,
    \qquad
    a\in\{0,\pi\},\quad 0<\theta<\pi,
    \label{eq:fixed_circle_nonzero_F}
\end{equation}
so that the $x$ component of the polarization has no additional zeros along the reflection-fixed circles away from the high-symmetry points $\theta=0,\pi$, then
\begin{equation}
    \left\langle w_1(\ell_a),[C_a]\right\rangle
    =
    \frac{
        1-\operatorname{sgn}
        \!\left[Z_T(a,0)Z_T(a,\pi)\right]
    }{2}
    \pmod2,
    \label{eq:w1_sign_formula}
\end{equation}
and therefore
\begin{equation}
    (-1)^{C_T}
    =
    \operatorname{sgn}\!\left[
        Z_T(0,0)
        Z_T(0,\pi)
        Z_T(\pi,0)
        Z_T(\pi,\pi)
    \right].
    \label{eq:chern_parity_product}
\end{equation}
By \cref{eq:high_symmetry_abelian}, each factor on the right-hand side is an ordinary current-parity characteristic function.
\end{theorem}

\begin{proof}
By \cref{lem:reflection},
\begin{equation}
    n_T(-\alpha,\theta)
    =
    \bigl(
        n_x(\alpha,\theta),
        -n_y(\alpha,\theta),
        n_z(\alpha,\theta)
    \bigr).
\end{equation}
Since complex conjugation leaves $\sigma_x$ and $\sigma_z$ unchanged and reverses the sign of $\sigma_y$,
\begin{equation}
    P_+\bigl(\tau(\alpha,\theta)\bigr)
    =
    \overline{P_+(\alpha,\theta)}.
    \label{eq:projector_real_symmetry}
\end{equation}
Thus complex conjugation defines an antilinear involution of $L_T$ covering $\tau$. In the terminology of Atiyah \cite{Atiyah1966}, $L_T$ is therefore a Real line bundle. On the fixed circles $C_0\sqcup C_\pi$, this antilinear symmetry reduces to ordinary complex conjugation within each fiber, whose fixed vectors form the real line bundles $\ell_0$ and $\ell_\pi$ introduced above.

To apply the corresponding result for Real curves, we equip the counting torus with its standard complex structure by writing \begin{equation} 
z=\theta+i\alpha, 
\qquad 
\mathbb T^2 \simeq \mathbb C/ \bigl(2\pi\mathbb Z+2\pi i\,\mathbb Z\bigr). 
\label{eq:counting_torus_complex} 
\end{equation} 
Under this identification, the reflection $\tau(\alpha,\theta)=(-\alpha,\theta)$ becomes complex conjugation, 
\begin{equation} 
z\longmapsto \bar z. 
\end{equation} 
Thus the counting torus is a compact Real curve, with real locus consisting precisely of the two fixed circles $C_0=\{\alpha=0\}$ and $C_\pi=\{\alpha=\pi\}$.

For a complex line bundle with this involutive structure over a compact Real curve with nonempty real locus, the parity of its degree is determined by the first Stiefel--Whitney classes of the real fixed-line bundles \cite{BiswasHuismanHurtubise2010}:
\begin{equation}
    \left\langle c_1(L_T),[\mathbb T^2]\right\rangle
    \equiv
    \left\langle w_1(\ell_0),[C_0]\right\rangle
    +
    \left\langle w_1(\ell_\pi),[C_\pi]\right\rangle
    \pmod2.
\end{equation}
Here
\[
    \left\langle c_1(L_T),[\mathbb T^2]\right\rangle=C_T
\]
is the first Chern number of $L_T$ \cite{MilnorStasheff1974,Husemoller1994}. Therefore \cref{eq:chern_sw_relation} follows.

It remains to determine the two Stiefel--Whitney numbers from the high-symmetry polarizations. On a fixed circle $\alpha=a$,
reflection symmetry gives $n_y=0$, and
\begin{equation}
    n_T(a,\theta)
    =
    \frac{
        \bigl(
            \sin\theta\,F_T(\cos a,\cos\theta),
            0,
            Z_T(a,\theta)
        \bigr)
    }{
        \|m_T(a,\theta)\|
    }.
    \label{eq:fixed_circle_texture}
\end{equation}
\cref{eq:fixed_circle_nonzero_F} implies that the
$x$ component can vanish on the fixed circle only at $\theta=0,\pi$, where the factor $\sin\theta$ vanishes. Thus the real eigenline can cross the $z$ axis only at these two high-symmetry points. Its Stiefel--Whitney number is therefore determined by whether the two corresponding $z$ polarizations have the same or opposite sign:
\begin{equation}
    (-1)^{
        \left\langle w_1(\ell_a),[C_a]\right\rangle
    }
    =
    \operatorname{sgn}
    \!\left[
        Z_T(a,0)Z_T(a,\pi)
    \right].
\end{equation}
This is equivalent to \cref{eq:w1_sign_formula}. Multiplying the contributions from $C_0$ and $C_\pi$ and using \cref{eq:chern_sw_relation} gives \cref{eq:chern_parity_product}.
\end{proof}

\begin{remark}
The congruence \cref{eq:chern_sw_relation} should not be confused with the familiar identity
\begin{equation}
w_2\!\left((L_T)_{\mathbb R}\right)=c_1(L_T)\pmod2,
\end{equation}
which concerns the underlying oriented real rank-two bundle \cite{MilnorStasheff1974}. Here reflection symmetry produces genuine real rank-one bundles on the fixed circles, and their $w_1$ classes encode the mod-$2$ Chern information.
\end{remark}

\begin{corollary}[Joint-parity criterion]
\label{cor:joint_parity}
Suppose the hypotheses of \cref{thm:chern_parity} hold and, in addition,
\begin{equation}
Z_T(0,\pi)>0,
\qquad
Z_T(\pi,0)>0.
\label{eq:other_parities_positive}
\end{equation}
Since $Z_T(0,0)=1$,
\begin{equation}
(-1)^{C_T}
= \operatorname{sgn}Z_T(\pi,\pi)
= \operatorname{sgn}\!\left[P_{\rm even}(T)-P_{\rm odd}(T)\right].
\label{eq:joint_parity_chern_mod2}
\end{equation}
Thus a sign reversal of the joint-current parity observable flips the parity of the Chern number.
\end{corollary}

This gives a precise, but not identical, connection to the Stiefel--Whitney description of parity counting in Chernyak and Sinitsyn \cite{ChernyakSinitsyn2010}. There, a first Stiefel--Whitney class characterizes a real eigenline encountered at the Abelian counting point $\chi=\pi$ during a cyclic driving protocol. Here the primary object is a complex counter-spin eigenline over a two-dimensional non-Abelian counting torus. Reflection symmetry restricts it to real eigenlines on fixed circles, whose $w_1$ data constrain $c_1$ modulo two. The common appearance of $\chi=\pi$ reflects a characteristic-class compatibility mediated by symmetry rather than an identification of the two constructions.

Note that while \cref{eq:joint_parity_chern_mod2} can be used to determine its parity, the integer value of Chern number still needs to be obtained from the geometric degree of the full map $n_T:\mathbb T^2\to\mathbb S^2$ the full two-dimensional spin texture. This distinction remains valid even when additional zeros of the transverse components occur away from the high-symmetry points.

The following proposition supplies an independent short-time constraint.

\begin{proposition}
\label{prop:weak_activity}
Let $r_{\max}$ be the largest escape rate of the underlying Markov chain. If
\begin{equation}
r_{\max}T<\log2,
\label{eq:short_time_condition}
\end{equation}
then $m_{T,z}(\alpha,\theta)>0$ for every $(\alpha,\theta)$ and therefore
\begin{equation}
C_T=0.
\end{equation}
\end{proposition}

\begin{proof}
For any initial state, the probability of no jump during $[0,T]$ is at least $p_0=e^{-r_{\max}T}$. No-jump trajectories leave the spin at $e_z$. Every other trajectory has final $z$ component at least $-1$, so uniformly in the counting angles
\[
m_{T,z}\ge p_0-(1-p_0)=2e^{-r_{\max}T}-1.
\]
Under \cref{eq:short_time_condition} this is positive. Hence the image of $n_T$ lies in the northern hemisphere and is null-homotopic.
\end{proof}

\section{Numerical results}
\label{sec:numerics}
We use the Markov network described in \cref{sec:model} and set the parameters as follows.
\begin{equation}
k_{1,+}=1.6,
\qquad k_{1,-}=0.55,
\qquad k_{2,+}=1.35,
\qquad k_{2,-}=0.65.
\label{eq:parameters}
\end{equation}
By \cref{eq:cycle_currents},
\begin{equation}
J_1=0.21,
\qquad
J_2=0.14.
\label{eq:thermodynamic_values}
\end{equation}

At $(\alpha,\theta,T)=(1.2,1.0,4)$, $50{,}000$ Gillespie trajectories \cite{Gillespie1977} give the mean spin shown in \cref{tab:validation}. The Euclidean discrepancy from the exact matrix-exponential result is $5.10\times10^{-3}$, consistent with the componentwise Monte Carlo errors.

\begin{table}[htbp]
\centering
\caption{Exact and Monte Carlo mean spin at $(\alpha,\theta,T)=(1.2,1.0,4)$.}
\label{tab:validation}
\begin{tabular}{lrrr}
\toprule
Component & Exact & Monte Carlo & Standard error\\
\midrule
$s_x$ & $0.294892$ & $0.295846$ & $0.002051$\\
$s_y$ & $-0.362790$ & $-0.362462$ & $0.002133$\\
$s_z$ & $0.146626$ & $0.151630$ & $0.002531$\\
\bottomrule
\end{tabular}
\end{table}

The simulated mean cycle counts are $0.83598$ and $0.56136$, compared with the stationary predictions $J_1T=0.84$ and $J_2T=0.56$. The explicit order test gives
\[
\norm{s_{AB}-s_{BA}}=0.686636,
\]
while the Frobenius norm of the rotation commutator is $1.474157$. Thus the spin counter records substantial path-order information without changing the ordinary currents.

At $T=4$ the exact mean-spin norm has a numerical global minimum
\begin{equation}
\min_{\T^2}\norm{m_T}
=4.1598623426\times10^{-3}
\label{eq:T4_gap}
\end{equation}
at $(\alpha,\theta)=(\pi,\pi)$ to numerical precision. The normalized field is therefore well defined numerically throughout the counting torus. \Cref{fig:texture} shows its three components and a sparse coordinate-net image on $\Sph^2$.

\begin{figure}[htbp]
\centering
\begin{subfigure}[t]{0.5\textwidth}
\centering
\includegraphics[width=\linewidth]{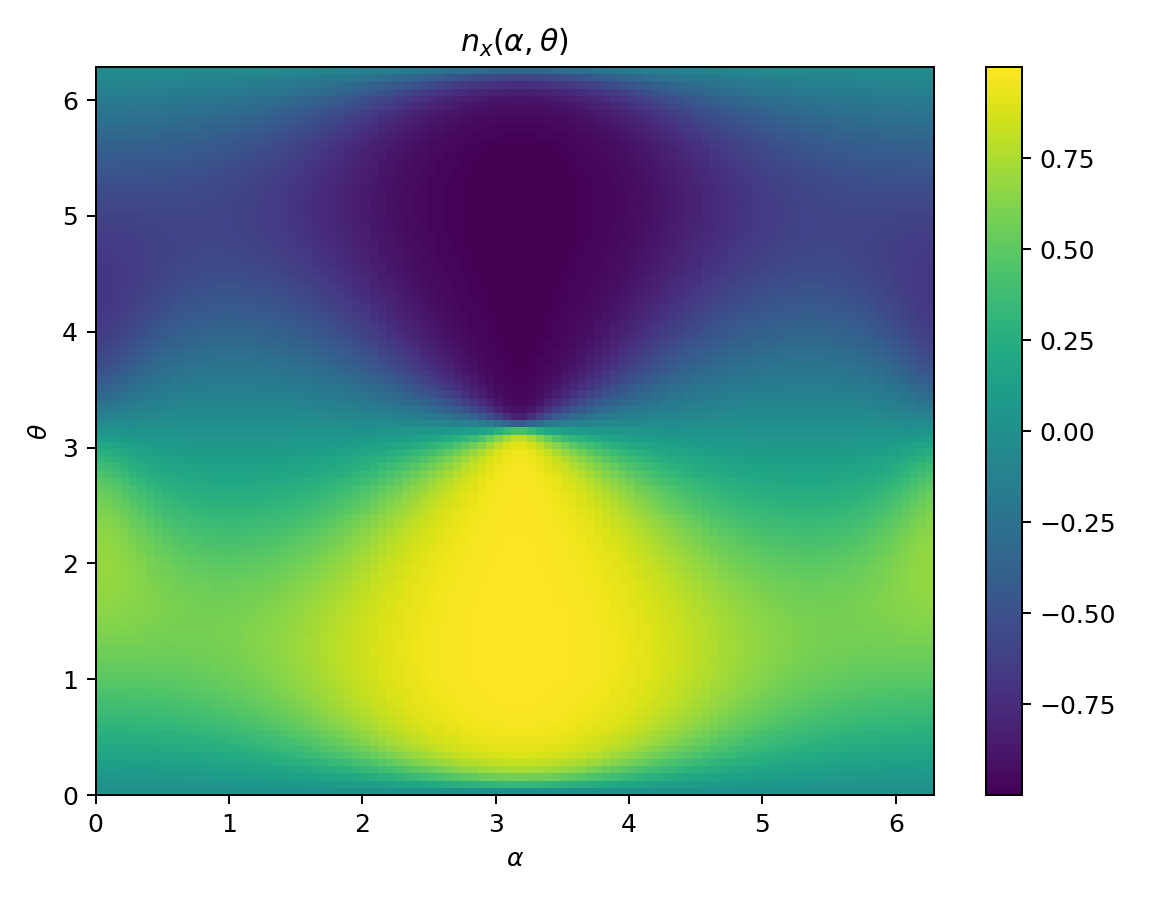}
\caption{$n_x(\alpha,\theta)$.}
\end{subfigure}\hfill
\begin{subfigure}[t]{0.5\textwidth}
\centering
\includegraphics[width=\linewidth]{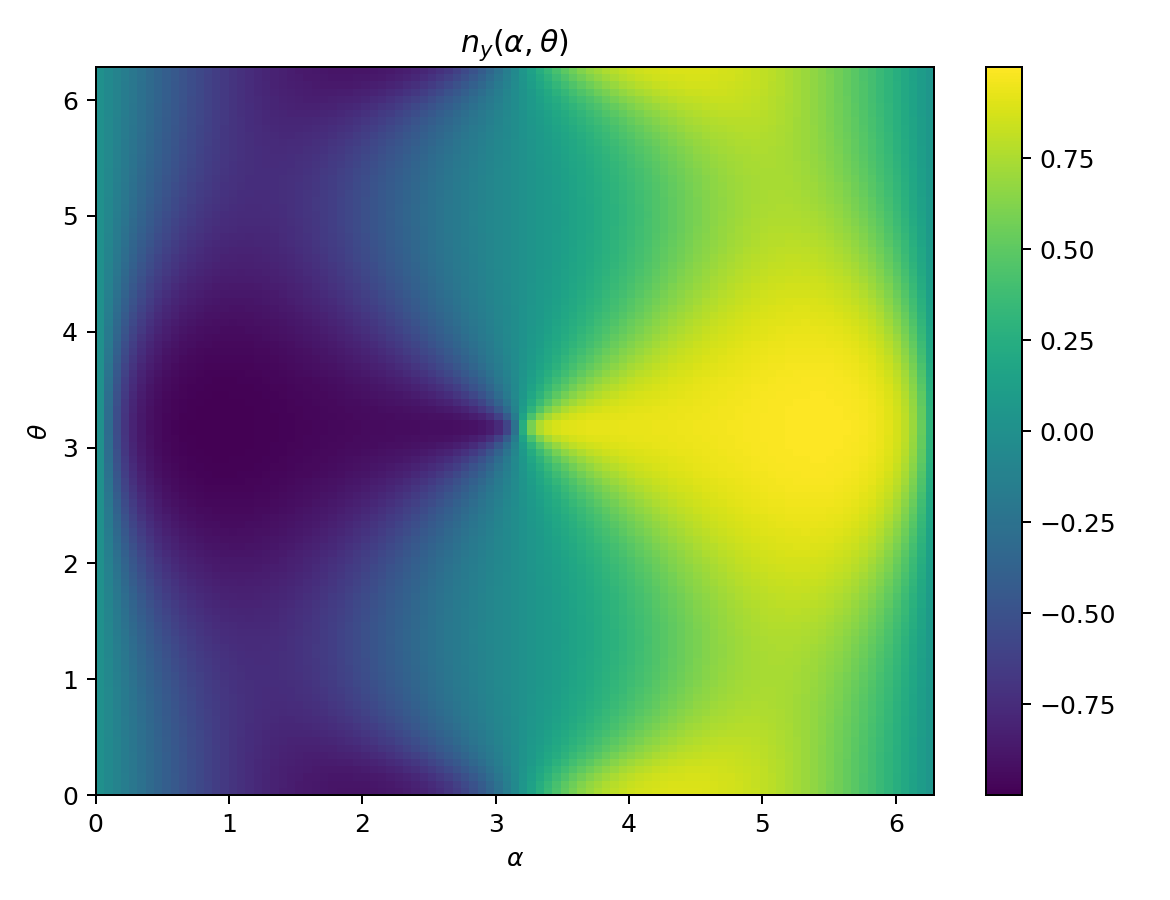}
\caption{$n_y(\alpha,\theta)$.}
\end{subfigure}

\medskip
\begin{subfigure}[t]{0.5\textwidth}
\centering
\includegraphics[width=\linewidth]{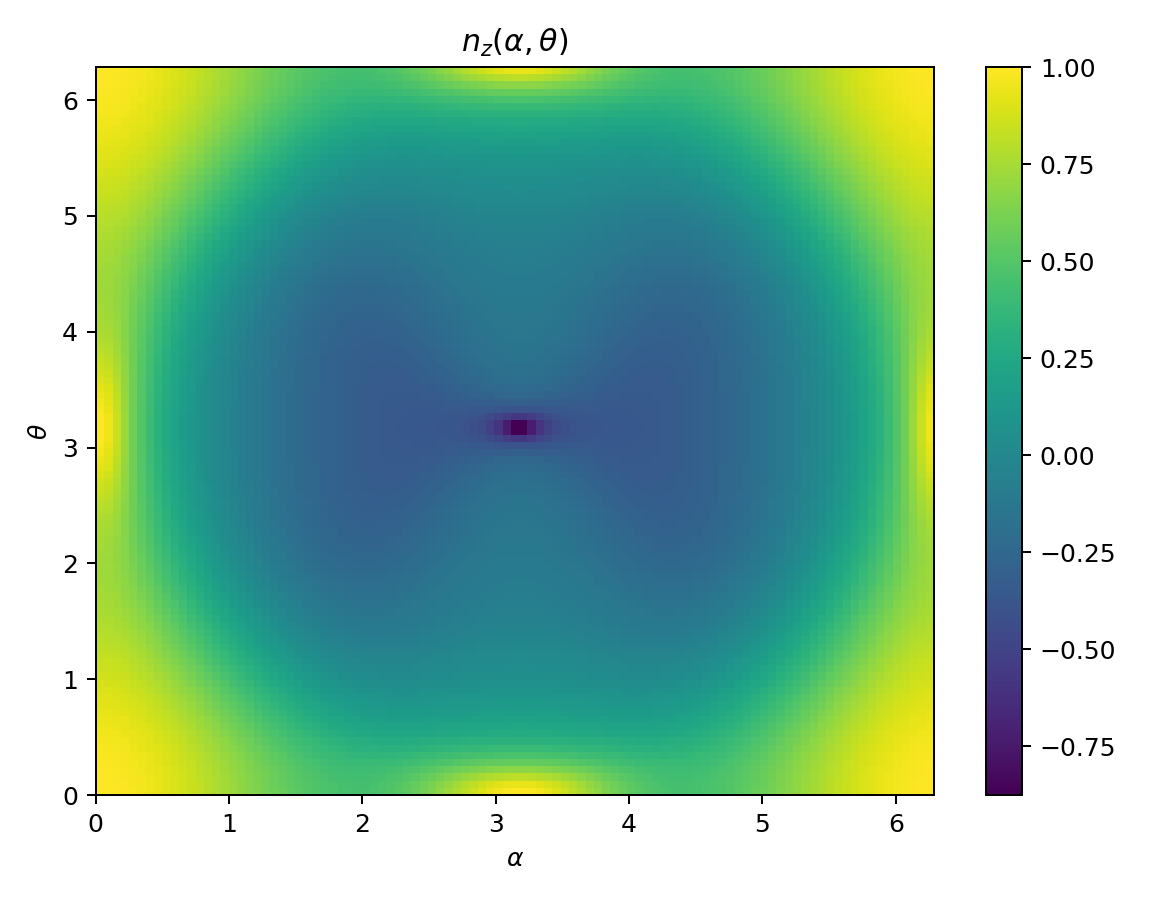}
\caption{$n_z(\alpha,\theta)$.}
\end{subfigure}\hfill
\begin{subfigure}[t]{0.5\textwidth}
\centering
\includegraphics[width=\linewidth]{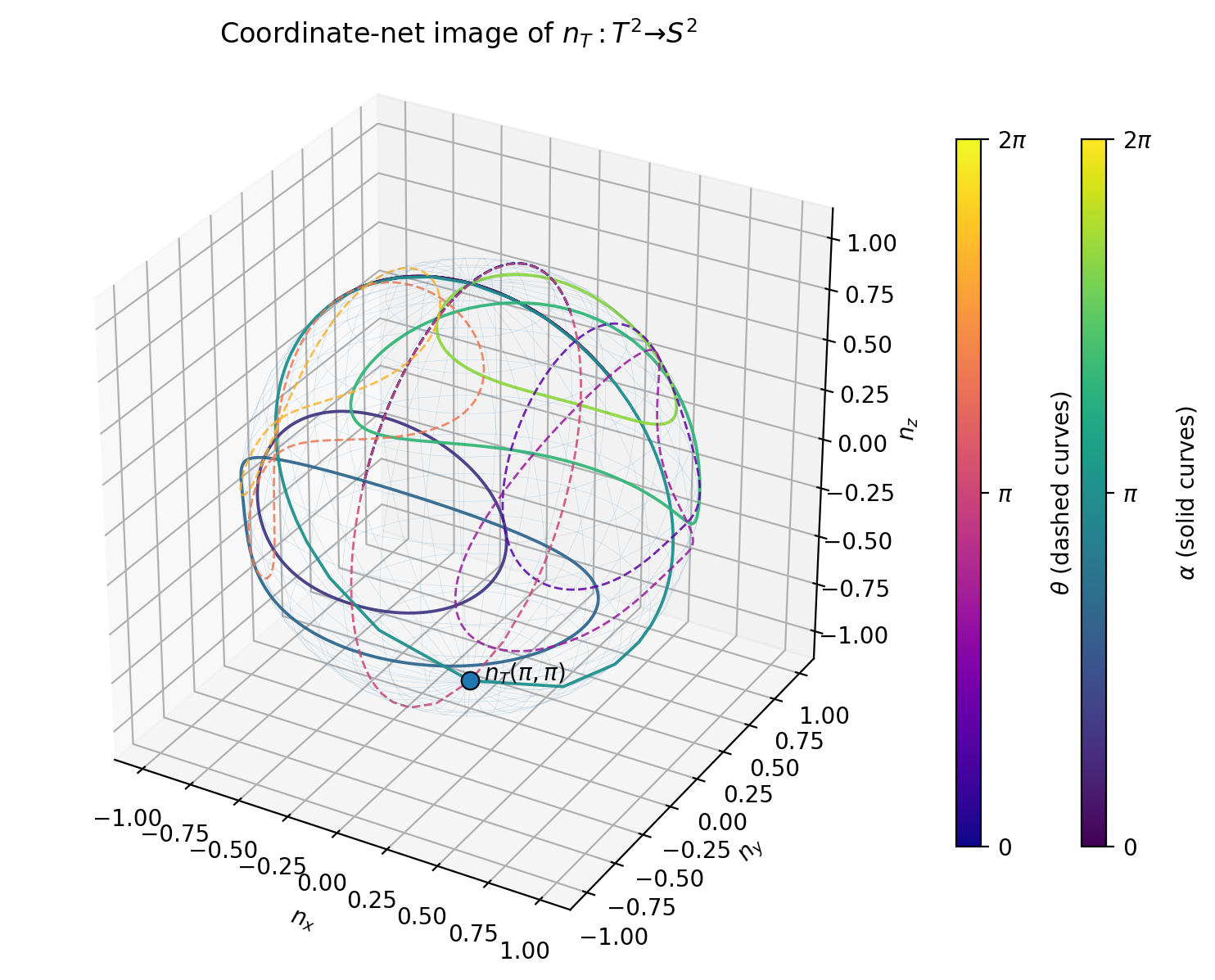}
\caption{Coordinate-net image on $\Sph^2$.}
\end{subfigure}
\caption{Normalized mean-spin texture at $T=4$. The component plots display the parameter-space structure, while the coordinate-net plot shows how the counting torus folds over the target sphere. The strong deformation near $(\pi,\pi)$ is associated with the small polarization norm there.}
\label{fig:texture}
\end{figure}

\subsection{Calculation of Chern number}
For the rank-one projector $P_+(\alpha,\theta)$ defined in \cref{eq:normalized_texture}, the Berry curvature is one-half the pullback of the standard area form on $\mathbb S^2$. Consequently, the curvature flux through a small triangle in the counting torus can be evaluated geometrically as one-half the oriented solid angle subtended by the image triangle on $\mathbb S^2$. Therefore, one can compute the Chern number using the lattice spherical-triangle formula \cite{BergLuscher1981}. For three unit vectors $a,b,c\in\mathbb S^2$, the oriented solid angle is
\begin{equation}
\Omega(a,b,c)
=
2\operatorname{atan2}\!\left(
a\cdot(b\times c),
1+a\cdot b+b\cdot c+c\cdot a
\right).
\label{eq:spherical_triangle}
\end{equation}
Dividing each periodic grid cell into two oriented triangles gives
\begin{equation}
C^{\mathrm{geom}}
=
\frac{1}{4\pi}
\sum_{\triangle}\Omega_\triangle.
\label{eq:lattice_chern}
\end{equation}
Thus the spherical-triangle construction is a discrete geometric
evaluation of the Chern curvature defined by $P_+$. At $T=4$, the result is $-1$ to machine precision at every tested resolution (\cref{tab:T4_convergence}), whereas the derivative-based integral converges more slowly because the normalized field varies sharply near $(\pi,\pi)$.

\begin{table}[htbp]
\centering
\caption{Grid refinement at $T=4$. The spherical-triangle degree is stable at $-1$; the central-difference integral converges toward the same integer.}
\label{tab:T4_convergence}
\begin{tabular}{rrrr}
\toprule
$N$ & $\min_{\rm grid}\norm{m_T}$ & $C_N^{\rm FD}$ & $C_N^{\rm geom}$\\
\midrule
21  & $1.21939\times10^{-2}$ & $-0.654185$ & $-1.000000000000$\\
31  & $8.76157\times10^{-3}$ & $-0.761263$ & $-1.000000000000$\\
41  & $7.15021\times10^{-3}$ & $-0.830225$ & $-1.000000000000$\\
61  & $5.70324\times10^{-3}$ & $-0.905900$ & $-1.000000000000$\\
81  & $5.09179\times10^{-3}$ & $-0.941712$ & $-1.000000000000$\\
101 & $4.77979\times10^{-3}$ & $-0.960653$ & $-1.000000000000$\\
\bottomrule
\end{tabular}
\end{table}
\FloatBarrier

\subsection{Repeated finite-time Chern transitions and their parity signature}
\label{sec:reentrant}

We next varied the observation time while keeping both the stochastic rates and the spin-counting rule fixed. Over the numerically resolved interval $0.02\le T\le14$, we evaluated the four high-symmetry polarizations
\[
Z_T(0,0),\qquad
Z_T(\pi,0),\qquad
Z_T(0,\pi),\qquad
Z_T(\pi,\pi),
\]
searched for the global minimum of $\|m_T\|$ over $\mathbb T^2$, and computed the geometric degree at representative times on an even $80\times80$ periodic grid.

The joint-parity polarization $Z_T(\pi,\pi)$ has four zeros,
\begin{align}
T_{c1}&=2.07416, &
T_{c2}&=5.40516, \nonumber\\
T_{c3}&=9.45918, &
T_{c4}&=12.1005.
\label{eq:critical_times}
\end{align}
In particular, the derivatives at the third and fourth roots are $Z'_T(\pi,\pi) = -3.06\times 10^{-6}$ and $1.08\times 10^{-7}$, respectively, confirming that the late-time crossings are not numerical zeros caused by double-precision roundoff.
At the refined roots,
\begin{align}
\|m_{T_{c1}}(\pi,\pi)\|
    &\approx 5.51\times10^{-17}, &
\|m_{T_{c2}}(\pi,\pi)\|
    &\approx 4.27\times10^{-18}, \nonumber\\
\|m_{T_{c3}}(\pi,\pi)\|
    &\approx 9.71\times10^{-19}, &
\|m_{T_{c4}}(\pi,\pi)\|
    &\approx 1.68\times10^{-19}.
\label{eq:critical_norms}
\end{align}
Refined minimization places all four polarization-gap closings at $(\pi,\pi)$. The other three high-symmetry polarizations remain positive over the resolved interval, whereas $Z_T(\pi,\pi)$ follows the sign sequence
\begin{equation}
    +,\;-,\;+,\;-,\;+.
\label{eq:Zpipi_sign_sequence}
\end{equation}

At this symmetry point the spin observable is exactly the ordinary joint-current parity characteristic function,
\begin{equation}
    Z_T(\pi,\pi)
    =
    \mathbb E\!\left[(-1)^{Q_1(T)+Q_2(T)}\right]
    =P_{\rm even}(T)-P_{\rm odd}(T).
\label{eq:Zpipi_parity_numeric}
\end{equation}
Hence each numerically resolved gap closing satisfies
\begin{equation}
    P_{\rm even}(T_{ci})=P_{\rm odd}(T_{ci}),
    \qquad i=1,\ldots,4.
\label{eq:critical_parity_balance}
\end{equation}
Over the same resolved interval, the factor $F_T$ remains nonzero on both reflection-fixed circles, while $Z_T(0,\pi)$ and $Z_T(\pi,0)$ remain positive. The hypotheses of \cref{cor:joint_parity} are therefore satisfied away from the four gap closings, and the theorem gives
\begin{equation}
(-1)^{C_T}=\operatorname{sgn}Z_T(\pi,\pi).
\label{eq:numeric_chern_parity}
\end{equation}
Thus the sign sequence in \cref{eq:Zpipi_sign_sequence} predicts the mod-$2$ sequence even--odd--even--odd--even. The integer Chern number and its sign are determined independently from the full two-dimensional texture, which, by the spherical-degree calculation below, is $0,-1,0,-1,0$. This separation is consistent with our discussion above: the parity statistic diagnoses the parity of the Chern number, whereas the non-Abelian spin texture carries the full integer invariant. Representative degree calculations on both sides of the four critical times are listed in \cref{tab:phase_degrees}, and \cref{fig:phase_scan} demonstrates a summary of the results. 

\begin{figure}[htbp]
\centering
\includegraphics[width=0.8\textwidth]{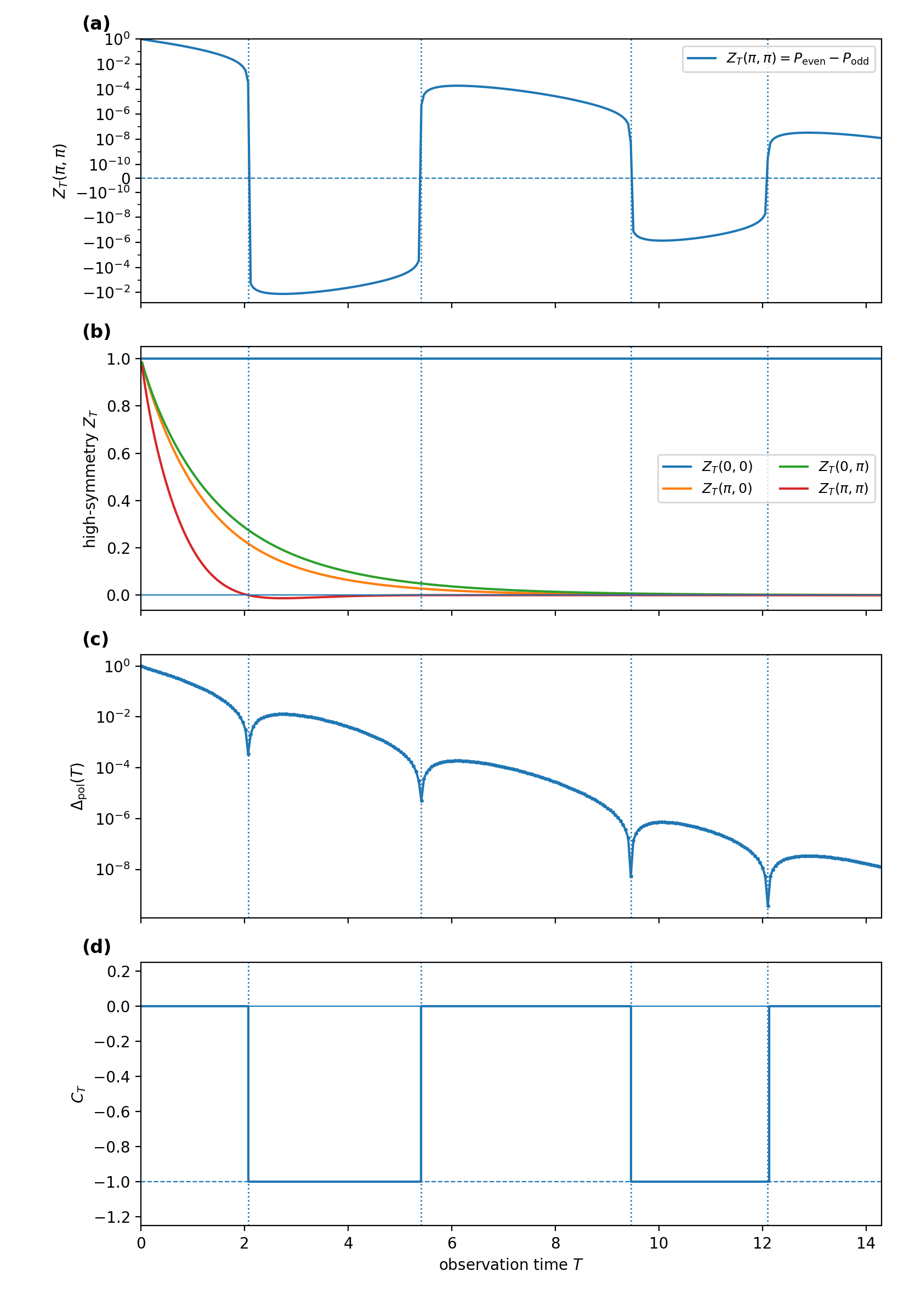}
\caption{
Finite-time topological phase scan and parity signature.
(a) Joint-current parity observable $Z_T(\pi,\pi)=P_{\rm even}-P_{\rm odd}$ on a symmetric logarithmic scale; dotted lines mark its four zero crossings.
(b) The four high-symmetry polarizations. The first three remain positive, while $Z_T(\pi,\pi)$ changes sign four times.
(c) Polarization gap $\Delta_{\rm pol}(T)$, which closes at the same four critical times.
(d) Chern number obtained independently from spherical triangulation. At each critical time the eigenline bundle is undefined because the mean polarization vanishes. The one-to-one correspondence between the Chern transitions and $P_{\rm even}=P_{\rm odd}$ is therefore a symmetry-resolved parity signature of the Chern-number transitions of the non-Abelian counter.
}
\label{fig:phase_scan}
\end{figure}

\begin{table}[htbp]
\centering
\caption{Representative geometric degree calculations across the four finite-time topological transitions.}
\label{tab:phase_degrees}
\begin{tabular}{rrr}
\toprule
$T$ & $\min_{\rm grid}\|m_T\|$ & $C_T^{\rm geom}$\\
\midrule
0.02000  & $9.676\times10^{-1}$ & $0$\\
2.04920  & $1.296\times10^{-3}$ & $0$\\
2.09913  & $1.212\times10^{-3}$ & $-1$\\
4.01467  & $4.068\times10^{-3}$ & $-1$\\
5.38019  & $1.763\times10^{-5}$ & $-1$\\
5.43012  & $1.657\times10^{-5}$ & $0$\\
9.43421  & $7.930\times10^{-8}$ & $0$\\
9.48414  & $7.374\times10^{-8}$ & $-1$\\
10.00667 & $7.212\times10^{-7}$ & $-1$\\
12.00400 & $1.158\times10^{-8}$ & $-1$\\
12.50333 & $2.776\times10^{-8}$ & $0$\\
14.00133 & $1.724\times10^{-8}$ & $0$\\
\bottomrule
\end{tabular}
\end{table}

Over the observed interval,
\begin{equation}
C_T=
\begin{cases}
0,  & 0.02\le T<T_{c1},\\
-1, & T_{c1}<T<T_{c2},\\
0,  & T_{c2}<T<T_{c3},\\
-1, & T_{c3}<T<T_{c4},\\
0,  & T_{c4}<T\le14,
\end{cases}
\label{eq:reentrant_phase}
\end{equation}
with $C_T$ undefined at $T=T_{ci}$. The short-time trivial regime is consistent with \cref{prop:weak_activity}. The nontrivial sector then appears, disappears, reappears, and disappears again without any change in the physical generator or the counter geometry. The transitions are generated solely by the evolving finite-time distribution of ordered stochastic trajectories.

The later-time gaps are very small, so the corresponding first-moment topology becomes increasingly sensitive to perturbations and numerical resolution even though the geometric degree remains integer-valued away from the refined zeros. Physically, this loss of polarization corresponds to repeated random rotations progressively dephasing the first moment of the counter.

\subsection{Comparison with conventional kinetic timescales}
\label{sec:timescales}

The repeated finite-time phase diagram raises a natural question: are the four critical times simply familiar kinetic timescales of the underlying five-state Markov process? We first compare them with the relaxation spectrum of the untwisted physical generator $L$. Its four nonzero eigenvalues are
\begin{equation}
\lambda_j=
-1.1305971458,\;
-2.9725989130,\;
-3.9954444198,\;
-4.3513595214,
\label{eq:markov_spectrum}
\end{equation}
so the associated relaxation times
$\tau_j=1/|\operatorname{Re}\lambda_j|$ are
\begin{equation}
\tau_j=
0.8844883465,\;
0.3364059630,\;
0.2502850484,\;
0.2298132331.
\label{eq:relaxation_times}
\end{equation}
Writing $\tau_{\rm rel}=0.8844883465$ for the slowest relaxation time, the four critical times occur at
\begin{equation}
\frac{T_{c1}}{\tau_{\rm rel}}=2.3450,\qquad
\frac{T_{c2}}{\tau_{\rm rel}}=6.1111,\qquad
\frac{T_{c3}}{\tau_{\rm rel}}=10.6945,\qquad
\frac{T_{c4}}{\tau_{\rm rel}}=13.6808.
\label{eq:critical_relaxation_ratios}
\end{equation}
The residual amplitudes of the slowest ordinary relaxation mode at these times are
\begin{equation}
e^{\lambda_1T_{ci}}
=
9.58\times10^{-2},\;
2.22\times10^{-3},\;
2.27\times10^{-5},\;
1.14\times10^{-6},
\qquad i=1,\ldots,4.
\label{eq:critical_mode_amplitudes}
\end{equation}
Thus the later reentrant transitions occur when ordinary
physical-state relaxation is already essentially complete. In
particular, the repeated changes of $C_T$ cannot be identified with successive relaxation modes of the Markov generator.

A second comparison uses stochastic activity. The stationary mean escape rate is
\begin{equation}
\bar r=\sum_i p_i^{\rm ss}r_i=2.49,
\end{equation}
so the four critical times correspond to mean accumulated jump counts
\begin{equation}
\bar rT_{c1}=5.165,\qquad
\bar rT_{c2}=13.459,\qquad
\bar rT_{c3}=23.553,\qquad
\bar rT_{c4}=30.130.
\label{eq:critical_activities}
\end{equation}
These values are not equally spaced, nor are the topological sectors selected by a fixed accumulated number of jumps.

We finally ask whether the first onset of the nontrivial value $C_T=-1$ simply coincides with the time at which complete traversals of the fundamental cycles become typical. Let
\begin{equation}
\mathcal T_{A^+}
=
\inf\{t:\;0\to1\to2\to0
\text{ has occurred consecutively in the jump sequence}\},
\label{eq:first_A_cycle}
\end{equation}
and define $\mathcal T_{A^-}$, $\mathcal T_{B^+}$, and $\mathcal T_{B^-}$ analogously. Unlike the chord counts in
\cref{eq:cycle_counts}, which record individual distinguished-edge crossings, these variables record the first occurrence of an entire three-edge directed word. Their distributions can be computed exactly by augmenting the physical chain with a finite pattern-progress variable and treating word completion as an absorbing event. Details are given in \cref{app:first_passage}.

\begin{table}[htbp]
\centering
\caption{
Exact first-passage statistics for complete directed cycle words, starting from the stationary physical-state ensemble with zero word progress. The probabilities at $T_{c1}$ and $T_{c2}$ bracket the first nontrivial Chern window.}
\label{tab:first_passage}
\begin{tabular}{lrrrr}
\toprule
Word & $\E[\mathcal T]$ & median
& $\Prob(\mathcal T\le T_{c1})$
& $\Prob(\mathcal T\le T_{c2})$\\
\midrule
$A^+:0\to1\to2\to0$ & $6.372$ & $4.724$ & $0.193$ & $0.558$\\
$A^-:0\to2\to1\to0$ & $139.647$ & $97.083$ & $0.008$ & $0.032$\\
$B^+:0\to3\to4\to0$ & $8.858$ & $6.471$ & $0.131$ & $0.428$\\
$B^-:0\to4\to3\to0$ & $73.556$ & $51.295$ & $0.015$ & $0.059$\\
\bottomrule
\end{tabular}
\end{table}
\FloatBarrier

At the first Chern transition, only about $19\%$ of trajectories have completed the full forward $A$ word and about $13\%$ the full forward $B$ word. In particular,
\begin{equation}
T_{c1}<\operatorname{med}(\mathcal T_{A^+}),
\qquad
T_{c1}<\operatorname{med}(\mathcal T_{B^+}).
\label{eq:Tc1_before_typical_cycle}
\end{equation}
The first appearance of the $C_T=-1$ sector therefore precedes the time at which a complete forward traversal of either fundamental cycle is typical. Conversely, the later transitions at $T_{c3}$ and $T_{c4}$ occur well after the forward-cycle median times and after ordinary state-space relaxation has effectively decayed. The repeated topological changes therefore cannot be reduced either to relaxation of the physical Markov state or to a single threshold for cycle completion.

Taken together, these comparisons support a more specific
interpretation of the finite-time Chern transitions. They reflect the changing ensemble of ordered stochastic histories accumulated by the non-Abelian counter. The ordinary physical dynamics, mean activity, and complete-cycle first-passage statistics provide natural kinetic reference scales, but none individually reproduces the sequence of four critical times.  In the present reflection-symmetric model, the same collective rearrangement has the simple parity signature $P_{\rm even}=P_{\rm odd}$ at each topological transition.

\section{Discussion}
\label{sec:discussion}

The finite-time transitions show that observation time can change the Chern number of the counter state even when the underlying stochastic process and the counting rule are fixed.  What evolves with $T$ is the distribution of ordered trajectory histories sampled by the noncommuting counter. At short times, no-jump and low-jump histories keep the mean spin close to its initial direction. As longer trajectories accumulate, different orderings of the two cycle rotations contribute with changing weights and reorganize the two-dimensional spin texture. The resulting texture can lose and recover nonzero degree even though the Markov generator, stationary state, and rotation rule remain unchanged.

The full integer $C_T$ is a property of this order-sensitive two-parameter spin texture. Reflection symmetry, however, makes its parity accessible through one-dimensional data on the two reflection-fixed circles. As established in \cref{thm:chern_parity}, the parity of the first Chern number is determined by the first Stiefel--Whitney classes of the corresponding real eigenline bundles. When the $x$-component of the polarization has no additional zeros along these circles away from the high-symmetry points, those Stiefel--Whitney numbers can in turn be read from the four high-symmetry polarizations. For the parameters studied here, three of these quantities remain positive, leaving
\[
    Z_T(\pi,\pi)=P_{\rm even}-P_{\rm odd}
\]
as the observable that distinguishes even from odd values of $C_T$. The integer value and its sign are nevertheless determined only by the full map $n_T:\mathbb T^2\to\mathbb S^2$.

This distinction clarifies the role of temporal ordering in the construction. The parity observable itself is Abelian. At $(\pi,\pi)$ the non-Abelian counter reduces exactly to the ordinary joint-current parity characteristic function. Away from the high-symmetry points, by contrast, the rotations associated with the two cycles do not commute, and trajectories with the same $(Q_1,Q_2)$ can produce different counter states. Thus the high-symmetry parity statistic does not reconstruct the non-Abelian counting distribution. Rather, reflection symmetry allows the modulo-two part of a Chern number defined by the order-sensitive two-dimensional family to be detected by an ordinary current statistic.

The parity crossings also admit an interpretation within conventional full counting statistics. Ren and Sinitsyn
\cite{RenSinitsyn2013} outlined the analogy between a stochastic generating function and an equilibrium partition function. For a counted current $Q$, the characteristic function $\mathcal Z_T(\chi)=\langle e^{i\chi Q}\rangle$ plays the role of a trajectory-space partition function, while $\log\mathcal Z_T$ is the corresponding cumulant-generating function and is analogous to a free energy. Zeros of the analytically continued generating function are therefore dynamical analogues of Lee--Yang zeros. In the present
two-current problem,
\[
    \mathcal Z_T(\pi,\pi)
    =
    Z_T(\pi,\pi)
    =
    P_{\rm even}(T)-P_{\rm odd}(T),
\]
so each of the four resolved parity balances is a time at which the Lee--Yang zero set of the two-current generating function passes through the joint-parity point $(\chi_1,\chi_2)=(\pi,\pi)$.

Reflection symmetry ties these particular zeros directly to the
counter-spin degeneracy. At the joint-parity point,
\[
    m_T(\pi,\pi)
    =
    \bigl(0,0,Z_T(\pi,\pi)\bigr),
\]
and hence
\[
    \mathcal Z_T(\pi,\pi)=0
    \quad\Longleftrightarrow\quad
    m_T(\pi,\pi)=0.
\]
At such a time the two eigenvalues of the averaged spin state become degenerate and the dominant eigenline ceases to be defined. For all four transitions resolved here, these symmetry-point zeros coincide with changes of Chern parity. The Lee--Yang language therefore gives an ordinary full-counting-statistics description of the same high-symmetry singularities that appear as gap closings of the non-Abelian counter.

The comparisons in \cref{sec:timescales} show that these
critical times cannot be identified with any single conventional
kinetic scale. The first transition occurs before completion of
either forward fundamental-cycle word is typical, whereas the later transitions occur after ordinary state-space relaxation has essentially decayed. The accumulated activities at the four
transitions are also not equally spaced. Relaxation, activity, and cycle-completion statistics therefore provide useful reference scales but do not by themselves reproduce the observed sequence of critical times. The transitions instead reflect collective changes in the finite-time ensemble of ordered histories sampled by the auxiliary counter.

A possible physical realization of the order-sensitive counter is an orientational degree of freedom driven by a stochastic environment. For example, different transitions in an internal chemical or conformational network could induce finite rotations of a molecular motor about different axes. Because rotations about distinct axes do not commute, the resulting orientation would depend on the temporal sequence of stochastic transitions rather than only on their integrated currents. The spin counter studied here can be viewed as a minimal passive version of such a system, and the finite-time Chern number then characterizes the global dependence of this averaged orientation on the two rotation parameters.

Several limitations point to natural extensions. The present
construction uses only the first moment of the spin-counting distribution, which becomes increasingly weakly polarized at long times. Higher representation moments, tensor observables, or the full probability measure on the group orbit may retain ordering information after the mean spin has strongly dephased. The counter is also passive; incorporating detector back-action would require a coupled stochastic model. Finally, the repeated reentrance has been demonstrated here for one minimal network.  Determining which graph structures, rate asymmetries, and non-Abelian representations support finite-time Chern-number transitions remains an open problem. The very small late-time gaps in \cref{tab:phase_degrees} also motivate a quantitative analysis of robustness to perturbations and finite sampling.

\section{Conclusion}
We introduced an order-sensitive spin counter for stochastic network trajectories in which crossings of different cycles generate noncommuting rotations. Its finite-time mean polarization defines, whenever the polarization gap is open, a complex eigenline bundle over a two-angle counting torus and hence a first Chern number. Reflection symmetry relates the parity of this Chern number to first Stiefel--Whitney invariants on the reflection-fixed circles and, under the stated condition on the fixed-circle polarization, to ordinary
current-parity characteristic functions. In the five-state
figure-eight model, varying observation time alone produces four gap closings and the reentrant sequence
\[
    C_T=0\longrightarrow-1\longrightarrow0
    \longrightarrow-1\longrightarrow0.
\]
Each resolved transition occurs at $(\pi,\pi)$, where $P_{\rm even}=P_{\rm odd}$, while the full integer values are obtained independently from the degree of the two-dimensional spin texture.

The example shows how a fixed stationary stochastic process can acquire distinct finite-time Chern numbers when probed by a counter that retains temporal ordering information discarded by ordinary current statistics.  Reflection symmetry makes the parity of this integer accessible through an Abelian observable, but the full invariant remains a property of the non-Abelian counting family. Extending the construction to larger networks, other group representations, and higher moments should clarify how broadly characteristic classes can organize order-sensitive information in stochastic trajectory ensembles.

\section{Acknowledgments}
The author thanks Nikolai A. Sinitsyn for proposing the non-Abelian spin-counting example and for insightful discussions that helped motivate this work.

\appendix
\section{Exact first-passage distribution for a directed cycle word}
\label{app:first_passage}

For completeness, we summarize the exact calculation underlying
\cref{tab:first_passage}. Let
\[
w=(e_1,e_2,e_3)
\]
be one of the directed three-edge words $A^\pm$ or $B^\pm$. To determine the first time at which $w$ occurs, we augment the physical Markov state $X_t\in\{0,\ldots,4\}$ with a progress variable
\[
r_t\in\{0,1,2\}.
\]
The value of $r_t$ records how much of the beginning of $w$ has most recently been matched. Specifically, $r_t=0$ means that no nonempty prefix of $w$ is currently matched, $r_t=1$ means that the most recent jump matches $e_1$, and $r_t=2$ means that the two most recent jumps match $e_1,e_2$. After each physical jump, $r_t$ is updated to the length of the longest suffix of the observed jump sequence that is also a prefix of $w$.

If $r_t=2$ and the next jump is $e_3$, then the complete three-edge word $w$ has occurred. We stop the first-passage process at this event and represent completion by an absorbing state. Before completion, the augmented process therefore has the $5\times3=15$ transient states
\[
(X,r)\in\{0,\ldots,4\}\times\{0,1,2\}.
\]

Let $p_w(t)\in\mathbb R^{15}$ denote the probability vector over these transient augmented states, with components
\begin{equation}
[p_w(t)]_{(i,r)} = \Pr\!\left(X_t=i,\,r_t=r,\,\mathcal T_w>t\right),
\label{eq:augmented_first_passage_prob}
\end{equation}
where $\mathcal T_w$ is the first time at which the complete word $w$ occurs. Thus $p_w(t)$ contains only probability mass from trajectories that have not yet completed $w$.

We initialize the physical process in its stationary distribution and assume that no part of the target word has already been matched. Hence the initial augmented probability vector is
\begin{equation}
[p_w(0)]_{(i,r)}
=
\begin{cases}
p_i^{\rm ss}, & r=0,\\
0, & r=1,2.
\end{cases}
\label{eq:first_passage_initial}
\end{equation}

Let $A_w$ be the $15\times15$ transient subgenerator governing
transitions among the augmented states for which $w$ has not yet been completed. Jumps that remain within the transient state space appear in $A_w$. Jumps that complete the word are instead collected in an absorption-rate vector $q_w\in\mathbb R^{15}$, where $(q_w)_{(i,r)}$ is the total rate at which the target word is completed from the augmented state $(i,r)$. In the column-vector convention used throughout the paper, the transient probability vector evolves according to
\begin{equation}
p_w(t)=e^{tA_w}p_w(0).
\label{eq:first_passage_transient_evolution}
\end{equation}

Since $p_w(t)$ contains precisely the probability mass of trajectories that have not yet completed the word, summing its components gives the survival probability
\begin{equation}
S_w(t)
=
\mathbf 1^\mathsf T e^{tA_w}p_w(0)
=
\Pr(\mathcal T_w>t).
\label{eq:first_passage_phase_type}
\end{equation}
The probability flux from the transient states into the absorbing completion state gives the first-passage density,
\begin{equation}
f_w(t)
=
q_w^\mathsf T e^{tA_w}p_w(0)
=
-\frac{d}{dt}S_w(t).
\label{eq:first_passage_density}
\end{equation}
The cumulative distribution is therefore
\begin{equation}
F_w(t)
=
\Pr(\mathcal T_w\le t)
=
1-S_w(t).
\end{equation}

Finally, the mean first-passage time is obtained by integrating the survival probability:
\begin{equation}
\E[\mathcal T_w]
=
\int_0^\infty S_w(t)\,dt
=
\mathbf 1^\mathsf T(-A_w)^{-1}p_w(0).
\label{eq:first_passage_mean}
\end{equation}
Thus the directed cycle-word first-passage statistics used in
\cref{sec:timescales} are computed exactly from the finite-dimensional augmented generator.

\bibliographystyle{unsrtnat}
\bibliography{references}

\end{document}